\documentclass[11pt]{article}
\usepackage{relate}
\usepackage{pgfplots}
\usepackage{pgfplotstable}
\pgfplotsset{compat=1.8}
\usepackage[aboveskip=8pt,belowskip=-4pt]{caption}
\usepackage{subcaption}
\usepackage{xspace} 
\usepackage{textcomp}
\usepackage{balance}  
\usepackage{xcolor}
\usepackage{amsmath, amsfonts}
\usepackage{algorithm}
\usepackage[margin=1in]{geometry}
\usepackage{algorithmicx}
\usepackage{float}
\usepackage{siunitx}
\usepackage{graphicx}
\usepackage{xspace}
\usepackage{hyperref}
\usepackage{thm-restate}
\usepackage{balance}  
\usepackage{booktabs} 

\usepackage{marginnote} 
\usepackage{url}
\usepackage{amssymb}
\usepackage{enumitem}
\graphicspath{{./fig/}}
\usepackage{amsthm}
\usepackage{mathtools}
\usepackage[noend]{algpseudocode}
\newtheorem{theorem}{Theorem}
\newtheorem{lemma}{Lemma}
\newtheorem{corollary}{Corollary}

\newtheorem{result}{Result}

\usepackage{comment}
\usepackage{color}
\usepackage{tikz}
\usetikzlibrary{fadings}

\input{defines}
\date{}

\begin{document}
\title{The Online Event-Detection Problem}
\author{
Michael A.~Bender\thanks{Department of Computer Science, 
Stony Brook University, Stony Brook, NY, 11794-2424 USA. Email:
\texttt{bender@cs.stonybrook.edu}.}
\and
Jonathan W. Berry\thanks{MS 1326, PO Box 5800, {Albuquerque}, NM, {87185} USA. Email: \texttt{\{jberry, caphill\}@sandia.gov}.}
\and 
Mart\'{\i}n Farach-Colton\thanks{Department of Computer Science, Rutgers University, Piscataway, NJ 08854 USA. Email: \texttt{farach@cs.rutgers.edu}.}
\and  
Rob Johnson\thanks{VMware Research, Creekside F
3425 Hillview Ave, {Palo Alto}, CA 94304 USA. Email: \texttt{robj@vmware.com}.}\vspace*{0.2pt}
\and Thomas M. Kroeger\thanks{MS 9011, PO Box 969, Livermore, CA 94551 USA. Email: \texttt{tmkroeg@sandia.gov}.}
\and
Prashant Pandey\thanks{Department of Computer Science, Carnegie Mellon University, 5000 Forbes Ave, Pittsburgh, PA 15213. Email: \texttt{ppandey2@cs.cmu.edu}.}
\and Cynthia A. Phillips\footnotemark[2]\vspace*{0.2pt} 
\and Shikha Singh\thanks{Department of Computer Science, Wellesley College, Wellesley, MA 02481-8203 USA.  Email: \texttt{shikha.singh@wellesley.edu}.} 
}

\maketitle
\sloppy

\begin{abstract}

	Given a stream $S = (s_1, s_2, \ldots, s_N)$, a \defn{$\phi$-heavy
	hitter} is an item $s_i$ that occurs at least $\phi N$ times in $S$.
	The problem of finding heavy-hitters has been extensively studied
	in the database literature.
%
	In this paper, we study a related problem.
	We say that there is a \defn{$\phi$-event at time~$t$} if $s_t$ occurs
	exactly $\lceil\phi N\rceil$ times in $(s_1, s_2, \ldots, s_{t})$.
	Thus, for each $\phi$-heavy hitter there is a single $\phi$-event, which
	occurs when its count reaches the \defn{reporting threshold}
	$T=\lceil\phi N\rceil$.
	We define the \defn{online event-detection problem} (\oedp) as: given $\phi$
	and a stream $S$, report all $\phi$-events as soon as they occur.

	Many real-world monitoring systems demand event detection where all events must
	be reported (no false negatives), in a timely manner, with no non-events reported
	(no false positives), and a low reporting threshold.  As a result,
	the \oedp requires a large amount of space ($\Omega(N)$ words)
	and
	 is not solvable in the streaming model or via standard sampling-based
	approaches.  
	
	Since \oedp requires large space, we focus on cache-efficient
	algorithms in the external-memory model.


	We provide algorithms for the \oedp that are within a log factor of optimal.
	Our algorithms are tunable: their parameters can be set to allow for
	bounded false-positives and a bounded delay in reporting.
	None of our relaxations allow false negatives since reporting all events
	is a strict requirement for our applications. 
	Finally, we show  
	improved results when the count of items in the input 
	stream follows a power-law distribution.

\end{abstract}

\section{Introduction}

Real-time monitoring of high-rate data streams, 
with the goal of detecting and preventing malicious events, is a critical component of defense systems for 
cybersecurity~\cite{Raza13,Sceller17,Yan09} and physical systems, such as water or power
distribution~\cite{Berry09,Kezunovic06,Litvinov06}.  In such a monitoring
system, changes of state are inferred from the stream elements.
Each detected/reported event triggers an intervention.
Analysts use more specialized tools to gauge the actual threat level.
Newer systems are even beginning to take defensive actions, such as
blocking a remote host, automatically based on detected
events~\cite{
MeinersPaNo10,GonzalezPaWe07}.  When
used in an automated system, accuracy (i.e., few false-positives and
no false-negatives) and timeliness of event detection are essential.

Motivated by these applications, we define and study the \defn{online event-detection problem} (\oedp).
Roughly speaking, the \oedp seeks to report all anomalous events (events that cross
a predetermined safety threshold) as soon as they occur in the input stream. The related problem of
finding the most frequent elements or heavy hitters in streams has been extensively
studied in the database literature~\cite{
CormodeMu05a, 
CormodeMu05b, 
AlonMaSz96,
CharikarChFa02, 
BravermanChIv16a, 
BravermanChIv16b, 
LarsenNeNg16,
BhattacharyyaDeWo16, 
DimitropoulosHuKi08, 
MankuMo02,BoseKrMo03,
DemaineOrMu02, 
BoseKrMo03,
BerindeIn10}.
More formally,
given a stream $S = (s_1, s_2, \ldots, s_N)$, a \defn{$\phi$-heavy
  hitter} is an element $s$ that occurs at least $\phi N$ times in
$S$.  
Here we focus on the
problem of finding $\phi$-events, where we say that there is
a \defn{$\phi$-event at time step~$t$} if 
$s_t$ occurs exactly $\lceil\phi N\rceil$ times in
$(s_1, s_2, \ldots, s_{t})$.  Thus for each $\phi$-heavy hitter there is a single
$\phi$-event which occurs when its count reaches the
\defn{reporting threshold} $T=\lceil\phi N\rceil$.

Formally, we define the \defn{online event-detection problem} (\oedp) as:
given stream $S = (s_1, s_2, \ldots, s_N)$, 
%
for each $i \in [1,N]$, 
report if there is a $\phi$-event at time $i$ before seeing $s_{i+1}$.
A solution to the online event-detection problem must report
\begin{enumerate}[noitemsep,nolistsep,label=(\alph*),itemindent=6pt,leftmargin=*]
\item\label{con1} all events\footnote{We drop the $\phi$ when it is understood.} (no \FNs) 
\item\label{con2} with no non-events  and no duplicates (no \FPs)
\item\label{con3} as soon as an element crosses the threshold (\OL).
\end{enumerate}
Furthermore, an online event detector must scale to
\begin{enumerate}[noitemsep,nolistsep,label=(\alph*),itemindent=6pt,leftmargin=*]
  \setcounter{enumi}{3}
\item small reporting thresholds $T$ and large $N$, i.e., very small $\phi$ (\Scalable).
\end{enumerate}

In this paper, we present algorithms for the \oedp.  We also give solutions
which relax conditions~(b) and~(c).  However, our solutions are motivated by
cybersecurity applications where (a) and (d) are strict requirements.
Next, we discuss how each of these conditions
relate to our approach and results. See~\secref{application} for more details
about the application that motivates the \oedp and its constraints.

\pparagraph{No false negatives}
We are motivated by monitoring systems for national security~\cite{FirehoseSite,AndersonPl15}.  
The events in this context have especially high
consequences so it is worth investing extra resources to detect them. 
We therefore do not allow false negatives (i.e., condition (a) is strict);
see~\secref{application} for more details. 
This rules out sampling-based approaches for the \oedp, which necessarily incur false negatives.

\pparagraph{Scalability}
Scalability (condition (d)) is essential in the broader context
of detecting anomalies in network streams, since anomalies are often small-sized events that develop slowly, appearing normal in the midst
of large amounts of legitimate traffic~\cite{mai2006sampled, Venkataraman05}.
As an example of the demands placed on event detection systems, the US
Department of Defense (DoD) and Sandia National Laboratories developed the
Firehose streaming benchmark suite~\cite{FirehoseSite,AndersonPl15} to measure
the performance of \oedp algorithms. In the FireHose benchmark, the reporting
threshold is preset to the representative value of $T=24$,
which translates to $\phi = 24/N = o(1)$ and thus a DoD benchmark enforces
condition (d).

\Scalable solutions to the \oedp require a large amount of space, ruling out the streaming
model~\cite{Babcock_etal_2002,Bar-Yossef_etal_2002}, where the available memory is small---usually just
$\polylog(N)$. In particular, streaming algorithms for the heavy-hitters problem assume $\phi > 1/\polylog(N)$ (all candidates 
must fit in memory).  
Even if some false positives are allowed,
as in the
$(\epsilon,\phi)$-heavy hitters problem\footnote{Given a stream of
  size $\Ns$ and $1/\Ns \leq \epsilon < \phi \leq 1$, report every
  item that occurs $\geq \phi N$ times and no item that occurs
  $\leq (\phi-\epsilon)$ times.  It is optional whether to report
  items with counts between $(\phi - \epsilon) \Ns$ and $\phi
  \Ns$.}, Bhattacharyya et al.~\cite{BhattacharyyaDeWo16} proved a lower bound of $(1/\epsilon) \log (1/\phi) +
(1/\phi) \log |\calU| + \log \log N$ bits. Thus the
space requirement is large when $\epsilon$ is small, as is the
case for \Scalable solutions where $\phi$ is small, since $\epsilon < \phi$. 


\pparagraph{Bounded false positives}
Our algorithms for the \oedp are tunable: parameters can be set to allow bounded false positives (relaxing condition (b)).
We show that allowing some false positives results in fewer I/Os per element.


Allowing \FPs does not lead to substantial
space savings.
If we allow 
$O(1+\beta t)$ false positives in a stream with $t$ true positives,
for any constant $\beta$, a bound of $\Omega(N\log N)$ bits
follows via a standard communication-complexity reduction from the
probabilistic indexing
problem~\cite{roughgarden2016communication,kushilevitz1997communication}.
Besides, as argued above, \Scalable solutions to the heavy-hitter problem require
large space even when \FPs are allowed.

%
%

\pparagraph{Bounded reporting delay}
The national-security monitoring systems we are interested in (see~\secref{application}) 
can tolerate a slight delay in reporting 
when the high-risk event
gives sufficient warning for
intervention. We show that allowing a bounded delay in reporting (relaxing condition (c)) 
allows us to circumvent the lower bounds on $\phi$ imposed
by our online solution. Thus, bounded
delay is especially desirable when we want our \oedp algorithm to
scale to arbitrarily small reporting thresholds.  

Finally, we note that in a security setting like ours, all 
events need to be detected in real-time to mitigate the associated risk. 
Thus streaming algorithms 
for the heavy-hitter problem 
that require multiple passes over the data are not applicable. 


\subsection*{Online Event Detection in External Memory}

In this paper, we make the large space requirement ($\Omega (N)$ words) 
of the \oedp more
palatable by shifting most
of the storage from expensive RAM to lower-cost external storage, such
as SSDs or hard drives. In particular, we give cache-efficient
algorithms for the \oedp in the external-memory model.
In the external-memory model, RAM has size $M$, storage has unbounded
size, and any I/O access to external memory transfers blocks of size $B$.
Typically, blocks are large, i.e., $B > \log N$~\cite{FrigoLePr12,
  AggarwalVi88}.

At first, it may appear trivial to detect heavy hitters using external
memory: we can store the entire stream, so what is there to solve?
And this would be true in an offline setting.  We could find all
events by logging the stream to disk and then sorting it.

The technical challenge to online event detection in external memory
is that searches are slow.  A straw-man solution is to maintain an
external-memory dictionary to keep track of the count of every item,
and to query the dictionary after each stream item arrives. But this
approach is bottlenecked on dictionary searches.  In a
comparison-based dictionary, queries take $\Omega(\log_B N)$ I/Os, and
there are many data structures that match this
bound~\cite{Comer79,BayerMc72,BrodalFa03b, BenderDeFa00}.  This yields
an I/O complexity of $O(N\log_B N)$.  Even if we use external-memory
hashing, queries still take $\Omega(1)$ I/Os, which still gives a
complexity of $\Omega(N)$ I/Os~\cite{IaconoPa12,ConwayFaPh18}.  Both
these solutions are bottlenecked on the latency of storage, which is
far too slow for stream processing.

Data ingestion is \textit{not} the bottleneck in external
memory.  Optimal external-memory dictionaries (including
write-optimized dictionaries such as
$B^\epsilon$-trees~\cite{BrodalFa03b,BenderFaJa15}, COLAs~\cite{BenderFaFi07},
xDicts~\cite{BrodalDeFi10}, buffered repository trees~\cite{BuchsbaumGoVe00}, write-optimized skip lists~\cite{BenderFa17}, log structured merge trees~\cite{ONeilCh96},
and optimal external-memory hash
tables~\cite{IaconoPa12,ConwayFaPh18}) can perform inserts and deletes extremely quickly. 
The fastest can index
using $O\!\left(\frac 1B \log{ \frac NM}\right)$ I/Os per stream element, which is far less than one
I/O per item.  In practice, this means that even a system with just a
single disk can ingest hundreds of thousands of items per second.  For
example, at SuperComputing 2017, a single computer was easily able to
maintain a $B^\epsilon$-tree~\cite{BrodalFa03b} index of all
connections on a 600 gigabit/sec network~\cite{Bender_etal_2018}.
The system could also efficiently
answer offline queries.  What the system could not do, however, was
detect events online.  

In this paper, we show how to achieve 
online (or nearly online) event detection for essentially the same
cost as simply inserting the data into a $B^\epsilon$-tree or other
optimal external-memory dictionary.

\subsection*{Results}

As our main result, we present an external-memory algorithm
that solves the \oedp, for $\phi$ that is sufficiently large, 
at an amortized I/O cost that is substantially cheaper than performing one query
for each item in the stream.
\begin{result}\resultlabel{onlineexact}
  Given a stream $S$ of size $\Ns$ and $\phi > 1/M +\Omega(1/N)$, the online
  event-detection problem can be solved at an amortized cost of
  $O\left(\left(\frac{1}{B}+\frac{M}{(\phi M-1)N}\right)\log \frac NM\right)$ I/Os per stream item.
\end{result}
To put this in context, suppose that
$\phi > 1/M \mbox{ and } (\phi > B/N \mbox{ or } N > MB)$. 
Then the I/O cost of solving the \oedp is $O\! \left( \frac{1}{B} \log \left({\frac{N}{M}}\right) \right)$,
which is only a logarithmic factor larger than the na\"\i{}ve scanning
lower bound.  In this case, we eliminate the query bottleneck and
match the data ingestion rate of $B^\epsilon$-trees.


Our algorithm builds on the classic Misra-Gries algorithm~\cite{Cormode08,MisraGries82}, and thus
supports its generalizations. 
In particular, similar to the $(\epsilon, \phi)$-heavy hitters problem,
our algorithm can also be relaxed 
so that items with frequency between $(\phi - \epsilon) N$ 
and $\phi N$ may be reported. 
Allowing false positives lowers the amortized I/O cost to  
$O\left(\left(\frac{1}{B}+\frac{M}{(\phi M-1)N}\right)\log \frac {1}{\epsilon M} \right)$;
see \thmref{immediate-MG}.
For the \oedp (i.e., no \FPs), we set $\epsilon = 1/N$.

Next, we show that, by allowing a bounded delay in reporting, we
can extend this result to arbitrarily small $\phi$.  Intuitively, we
allow the reporting delay for an event $s_t$ to be proportional to the
time it took for the element $s_t$ to go from 1 to $\phi\Ns$ occurrences.
More formally, for a $\phi$-event $s_t$, define the
\defn{flow time} of $s_t$ to be $F_t = t-t_1$, where $t_1$ is the time
step of $s_t$'s first occurrence.  We say that an event-detection
algorithm has \defn{time stretch $1+\alpha$} if it reports each event
$s_t$ at or before time $t + \alpha F_t=t_1+(1+\alpha) F_t$.
\begin{result}\resultlabel{timelyexact}
  Given a stream $S$ of size $\Ns$ and $\alpha>0$, the \oedp can be
  solved for any $\phi \geq 1$ with time stretch $1+\alpha$ at an amortized cost of
  $O\left(\frac{\alpha+1}{\alpha}\frac{\log N/M}{B}\right)$ I/Os per stream item.
\end{result}

For constant $\alpha$, this is asymptotically as fast as simply
ingesting and indexing the
data~\cite{BrodalFa03b,BenderFaFi07,BuchsbaumGoVe00}. 
This algorithm can also be relaxed to allow false positives
and achieve an improved I/O complexity. Thus, this 
result yields an almost-online
solution to the $(\epsilon,\phi)$-heavy hitters problem for
arbitrarily small $\epsilon$ and $\phi$; see~\thmref{timestretch}.


Finally, we consider input distributions where the count
of items is drawn from a power-law distribution.
Berinde et al.~\cite{BerindeIn10} show that the
Misra-Gries algorithm gives improved guarantees
for the heavy-hitter problem when 
the input follows a Zipfian distribution
with exponent $\alpha >1$. If the item counts in the stream
follow a Zipfian distribution with exponent $\alpha$
if and only if they
follow a power-law distribution with exponent $\theta = 1+1/\alpha$~\cite{adamic2008zipf}.\footnote{Zipf and power-law
are often used interchangeably in the literature, however, they are different
ways to model the same phenomenon; see~\cite{adamic2008zipf}
and~\secref{powerlaw}
for details.}
%
%
%
%
As our algorithms are based
on Misra-Gries, we automatically get the same improvements 
when the power-law exponent $\theta \leq 2$ (i.e., $\alpha >1$).

We design a data structure for the \oedp problem that
supports a smaller threshold $\phi$ than in~\resultref{onlineexact} and achieves a better I/O
complexity when the count of items in the stream
follow a power-law distribution with exponent
$\theta > 2+1/(\log_2 (N/M))$.
%
For a representative specification of 1TB hard drive
and 32GB RAM, our algorithm is performant for 
Zipfian distributions with $\alpha \leq 0.94$,
a range that is frequently observed in practical data~\cite{ClausetShNe09, BerindeIn10, newman2005power,BreslauCa99,adamic2008zipf}.
For instance, the number of connections to the internet backbone
at the autonomous-system level follow a Zipfian 
distribution with exponent $\alpha = 0.8$~\cite{adamic2008zipf}.


\begin{result}\resultlabel{powerlaw}
  Given a stream $S$ of size $\Ns$, where the count of items 
follows a power-law distribution with exponent $\theta > 1$, and
$\phi = \Omega(\gamma/N)$, where 
  $\gamma = 2 {({N}/{M})}^{\frac{1}{\theta-1}}$, the
  \oedp can be solved at an amortized I/O
  complexity
  $O\! \left( \left(\frac{1}{B} + \frac{1}{{(\phi N - \gamma)}^{\theta-1}} \right) \log \frac {N}{ M} \right)$ per stream item.
\end{result}
%
In contrast to the worst-case
solution (\resultref{onlineexact}), \resultref{powerlaw} allows 
thresholds $\phi$ smaller than $1/M$
and an improved I/O complexity 
when the power-law exponent $\theta > 2 + 1/(\log_2 (N/M))$.
(This is because $\gamma/N < 1/M$ in this case; see~\secref{powerlaw} for details.) 

\section{Preliminaries} 
\seclabel{misra-gries}


This section reviews the Misra-Gries heavy-hitters algorithm~\cite{MisraGries82},
a building block of our algorithms in~\secref{external-MG} and~\secref{timestretch}.

%
\pparagraph{The Misra-Gries frequency estimator}
The Misra-Gries (MG) algorithm estimates the frequency of items in a stream.
Given an estimation error bound $\epsilon$ and a stream $S$ of $\Ns$
items from a universe $\calU$, the MG algorithm uses a single pass
over $S$ to construct a table $\calC$ with at most  $\lceil 1/\epsilon
\rceil$ entries.  Each table entry is an item $s \in \calU$ with a count, denoted $\calC[s]$.
For each $s \in \calU$ not in table $\calC$, let $\calC[s] = 0$.  Let $f_s$ be the number of occurrences
of item $s$ in stream $S$.  The MG algorithm guarantees that $\calC[s]\leq f_s< \calC[s]+\epsilon \Ns$ for all $s \in \calU$.


MG initializes $\calC$ to an empty table and then processes the items
in the stream one after another as described below.  For each $s_{i}$ in $S$,
\begin{itemize}[noitemsep,nolistsep]
\item If $s_i\in \calC$, increment counter $\calC[s_i]$.
\item If $s_i\not\in \calC$ and $|\calC| < \lceil 1 / \epsilon \rceil$, insert $s_i$ into $\calC$ and 
  set $\calC[s_i]\leftarrow1$.
\item If $s_i \notin \calC$ and $|\calC| = \lceil 1 / \epsilon \rceil$,
  then for each $x \in \calC$ decrement $\calC[x]$ and delete its entry
  if $\calC[x]$ becomes 0.
\end{itemize}

We now argue that $\calC[s]\leq f_s< \calC[s]+\epsilon \Ns$.  We
have $\calC[s] \leq f_s$ because $\calC[s]$ is
incremented only for an occurrence of $s$ in the stream. MG underestimates counts
only through the decrements in the third condition above. 
This step decrements $\lceil 1/\epsilon \rceil + 1$ counts at once: the
item $s_i$ that caused the decrement, since it is never added to
the table, and each item in the table. There can be at most
$\lfloor \Ns / \lceil 1/\epsilon + 1\rceil \rfloor < \epsilon \Ns$ executions of this
decrement step in the algorithm. Thus, $f_s < \calC[s] + \epsilon \Ns$.

\pparagraph{The $(\epsilon, \phi)$-heavy hitters problem}
The MG algorithm can be used to solve the
\defn{$(\epsilon,\phi)$-heavy hitters problem}, which requires us to report
all items $s$ with $f_s\geq \phi\Ns$ and not to report any item
$s$ with $f_s\leq (\phi-\epsilon)\Ns$.  Items that occur strictly between
$(\phi-\epsilon)\Ns$ and $\phi \Ns$ times in $S$ are neither required
nor forbidden in the reported set. 

To solve the problem, run the MG algorithm on the stream with  
error parameter $\epsilon$.  Then
iterate over the set $\calC$ and report any item $s$ with $\calC[s]
> (\phi - \epsilon)\Ns $.  Correctness follows from  1)
if $f_s\leq (\phi - \epsilon)\Ns$, then
$s$ will not be reported,
since $\calC[s] \leq f_s \leq (\phi
- \epsilon)\Ns$, and 2) if $f_s\geq \phi
\Ns$, then $s$ will be reported, since
$\calC[s] > f_s - \epsilon \Ns \geq \phi \Ns - \epsilon
\Ns$.

\pparagraph{Approximate online-event detection}
Analogous to the $(\epsilon, \phi)$-heavy hitters
problem, we define the \defn{approximate \oedp} as: 
\begin{itemize}[nolistsep]
\item Report all $\phi$-events $s_t$ at time $t$,
\item Do not report any item $s_i$ with count at most $(\phi-\epsilon)\Ns$
\item Items with count greater than $(\phi-\epsilon)\Ns$ and less than $\phi N$
are neither required nor forbidden from being reported.
\end{itemize}

All the errors with respect to \oedp in the $(\epsilon, \phi)$-heavy hitters
problem and the approximate \oedp are \defn{false positives}, that is,
non-events (items 
with frequency between $(\phi - \epsilon)N$ and $\phi N)$
that get reported as $\phi$-events. No false negatives
are allowed as all $\phi$-heavy hitters and $\phi$-events must be reported. In the rest 
of the paper, the term error only refers to false-positive errors.




\pparagraph{Space usage of the MG algorithm}
For a frequency estimation error of $\epsilon$, Misra-Gries uses
$O(\lceil 1/\epsilon \rceil)$ words of storage, assuming each stream
item and each count occupy $O(1)$ words.  

Bhattacharyya et
al.~\cite{BhattacharyyaDeWo16} showed that, by using hashing,
sampling, and allowing a small probability of error, Misra-Gries can
be extended to solve the $(\epsilon,\phi)$-Heavy Hitters problem using
$1/\phi$ slots that store counts and an additional $(1/\epsilon) \log
(1/\phi) + \log \log \Ns$ bits, which they show is optimal.

For the exact $\phi$-hitters problem, that is, for $\epsilon =1/\Ns$,
the space requirement is large---$\Ns$ slots.  Even
the optimal algorithm of Bhattacharyya uses $\Omega(\Ns)$ bits of
storage in this case, regardless of $\phi$.  

%



\section{External-Memory Misra-Gries and Online Event Detection}
\seclabel{external-MG}

In this section, we design an efficient external-memory version of the
core Misra-Gries frequency estimator.  This immediately gives 
an efficient external-memory algorithm for the $(\epsilon,\phi)$-heavy
hitters problem.  We then extend our external-memory Misra-Gries
algorithm to support I/O-efficient immediate event reporting, e.g.,  for 
online event detection.

When $\epsilon=o(1/M)$, then simply running the standard Misra-Gries
algorithm can result in a cache miss for every stream element,
incurring an amortized cost of $\Omega(1)$ I/Os per element.  Our
construction reduces this to $O(\frac{\log (1/(\epsilon M))}{B})$,
which is $o(1)$ when $B=\omega\big(\log \big(\frac{1}{\epsilon M}\big)\big)$.

\subsection{External-memory Misra-Gries}\subseclabel{external-MG-subsec}\seclabel{external-MG-subsec}

Our external-memory Misra-Gries data structure is a sequence of Misra-Gries
tables, $\calC_0, \ldots, \calC_{L-1}$, where $L = 1+\lceil\log_r (1/(\epsilon
M))\rceil$ and $r ~(>1)$ is a parameter we set later. The size of the table
$\calC_i$ at level $i$ is $r^iM$, so the size of the last level is at least
$1/\epsilon$.

Each level acts as a Misra-Gries data structure.  Level 0 receives the input
stream.  Level $i > 0$ receives its input from level $i-1$, the level above.
Whenever the standard Misra-Gries algorithm running on the table $\calC_i$ at
level $i$ would decrement a item count, the new data structure decrements that
item's count by one on level $i$ and sends one instance of that item to the level
below ($i+1$).

The external-memory MG algorithm processes the input stream by inserting each
item in the stream into $\calC_0$.  To insert an item $x$ into level $i$, do
the following: \begin{itemize}[noitemsep,nolistsep]
\item If $x \in \calC_i$, then increment $\calC_i[x]$.  \item If $x \notin
\calC_i$, and $|\calC_i| \leq r^iM-1$, then $\calC_i[x] \gets 1$.  \item If $x
\notin \calC_i$ and $|\calC_i| = r^iM$, then, for each $x'$ in $\calC_i$,
decrement $\calC_i[x']$;  remove it from $\calC_i$ if $\calC_i[x']$ becomes 0.
If $i < L-1$, recursively insert $x'$ into $\calC_{i+1}$.  \end{itemize}

 We call the process of decrementing the counts of all the items at
  level $i$ and incrementing all the corresponding item counts at level
  $i+1$ a \defn{flush}. 

\pparagraph{Correctness} We first show that the external-memory MG algorithm
still meets the guarantees of the Misra-Gries frequency estimation
algorithm.  In fact, we
show that every prefix
of levels $\calC_0,\ldots,\calC_j$ is a Misra-Gries frequency
estimator, with the accuracy of the frequency estimates increasing with $j$.

\begin{lemma}
  Let $\widehat{\calC}_j[x] = \sum_{i=0}^{j}\calC_i[x]$ (where
  $\calC_i[x]=0$ if $x\not\in\calC_i$).  Then, the following holds:
\begin{itemize}
\item  $\widehat{\calC}_j[x]\leq f_x < \widehat{\calC_j}[x] + ({\Ns}/({r^jM}))$, and,
\item  $\widehat{\calC}_{L-1}[x]\leq f_x <
  \widehat{\calC}_{L-1}[x] + \epsilon\Ns$.
\end{itemize}
\end{lemma}
\begin{proof}
  Decrementing the count for an element $x$ in level
  $i<j$ and inserting it on the next level does not change
  $\widehat{\calC}_j[x]$.  This means that $\widehat{\calC}_j[x]$
  changes only when we insert an item $x$ from the input stream into
  $\calC_0$ or when we decrement the count of an element in level $j$.
  Thus, as in the original Misra-Gries algorithm, $\calC_j[x]$ is only
  incremented when $x$ occurs in the input stream, and is decremented
  only when the counts for $r^jM$ other elements are
  also decremented.  Following the same arguments as the MG algorithm, this
  is sufficient to establish the first inequality.
  The second inequality follows from the first, and the fact that
  $r^{L-1}M\geq 1/\epsilon$.
\end{proof}

\pparagraph{Heavy hitters}
Since our external-memory Misra-Gries data structure matches the
original Misra-Gries error bounds, it can be used to solve the
$(\epsilon,\phi)$-heavy hitters problem when the regular Misra-Gries
algorithm requires more than $M$ space.  First, insert each element of
the stream into the data structure.  Then, iterate over the sets
$\calC_i$ and report any element $x$ with counter
$\widehat{\calC}_{L-1}[x] > (\phi - \epsilon)\Ns $.

\pparagraph{I/O complexity}
We now analyze the I/O complexity of our external-memory Misra-Gries
algorithm.  For concreteness, we assume each level is implemented as a
B-tree, although the same basic algorithm works with sorted arrays
(included with fractional cascading from one level to the next, similar to
cache-oblivious lookahead arrays~\cite{BenderFaFi07}) or hash tables with
linear probing and a consistent hash function across levels (similar to
cascade filters~\cite{BenderFaJo11}).

\begin{lemma}
For a given $\epsilon \geq 1/N$, the amortized I/O cost of 
insertion in the external-memory Misra-Gries data structure is $O(\frac 1B \log {\frac {1}{\epsilon M}})$. 
\end{lemma} 
\begin{proof}
   Recall that the process of decrementing the counts of all the items at
  level $i$ and incrementing all the corresponding item counts at level
  $i+1$ is a flush. A flush can be implemented by rebuilding the
  B-trees at both levels, which can be done in $O(r^{i+1}M/B)$ I/Os.

  Each flush from level $i$ to level $i+1$
  moves $r^i M$ stream elements down one level, so the amortized cost
  to move one stream element down one level is $O(\frac{r^{i+1}M}{B} /
  (r^iM)) = O(r/B)$ I/Os.
  
  Each stream element can be moved down at most $L$ levels. Thus, the
  overall amortized I/O cost of an insert is $O(rL/B) = O((r/B) \log_r
	(1/(\epsilon M)))$, which is minimized at $r = e$.
\end{proof}
When no false positives are allowed, that is, $\epsilon = 1/\Ns$,
the I/O complexity of the external-memory MG algorithm is
$O(\frac 1B \log{ \frac NM})$.


\subsection{Online event-detection}
\seclabel{immediate-MG}
We now extend our external-memory Misra-Gries data structure to solve
the online event-detection problem.  In particular, we show that for  a
threshold $\phi$ that is sufficiently large, we can report
$\phi$-events as soon as they occur.

A first attempt to add immediate reporting to our external-memory Misra-Gries
algorithm is to compute $\widehat{\calC}_{L-1}[s_i]$ for each
stream event $s_i$ and report $s_i$ as soon as
$\widehat{\calC}_{L-1}[s_i]>(\phi-\epsilon)\Ns$.  However, this 
requires querying $\calC_i$ for $i=0,\ldots,L-1$ for every stream
item and can cost up to $O(\log(1/\epsilon M))$ I/Os per stream
item.

We avoid these expensive queries by using the properties of the in-memory
Misra-Gries frequency estimator $\calC_0$.
If $\calC_0[s_i] \leq
(\phi-1/M)\Ns$, then we know that $f_{s_i}\leq \phi\Ns$ and we therefore
do not have to report $s_i$, regardless of the count for $s_i$ in the lower levels on
disk of the external-memory data structure.

\pparagraph{Online event-detection in external memory}
We modify our external-memory Misra-Gries algorithm to support
online event detection as follows.  Whenever we increment
$\calC_0[s_i]$ from a value that is at most $(\phi-1/M)\Ns$ to a value
that is greater than $(\phi-1/M)\Ns$, we compute $\widehat{\calC}_{L-1}[s_i]$ and
report $s_i$ if
$\widehat{\calC}_{L-1}[s_i]=\lceil(\phi-\epsilon)\Ns\rceil$.  For each
entry $\calC_0[x]$, we store a bit indicating whether we have
performed a query for $\widehat{\calC}_{L-1}[x]$.  As in our basic
external-memory Misra-Gries data structure, if the count for an entry
$\calC_0[x]$ becomes 0, we delete that entry. This means we might
query for the same item more than once if its in-memory count
crosses the $(\phi-1/M)\Ns$ threshold, it gets removed from
$\calC_0$, and then its count crosses the $(\phi-1/M)\Ns$
threshold again.  
As we will see below, this has no affect
on the overall I/O cost of the algorithm.\footnote{It is possible
to prevent repeated queries for an item but we allow it
as it does not hurt the asymptotic performance.}

In order to avoid reporting the same item more than once, we can
store, with each entry in $\calC_i$, a bit indicating whether that item
has already been reported.  Whenever we report a item $x$, we set the
bit in $\calC_0[x]$.  Whenever we flush a item from level $i$
to level $i+1$, we set the bit for that item on level $i+1$ if it is
set on level $i$.  When we delete the entry for a item that has the bit
set on level $L-1$, we add an entry for that item on a new level
$\calC_L$.  This new level contains only items that have already been
reported.  When we are checking whether to report a item during
a query, we stop checking further and omit
reporting as soon as we reach a level where the bit is set.
None of these changes affect the I/O
complexity of the algorithm.

\pparagraph{I/O complexity}
We assume that computing
$\widehat{\calC}_{L-1}[x]$ requires $O(L)$ I/Os.  This is true if the
levels of the data structure are implemented as sorted arrays with
fractional cascading.

We first state the result for the approximate version of the online event-detection problem
that allows elements with frequency between $(\phi  - \epsilon) N$ and $\phi N$ to be
reported as false positives. 

Then, we set $\epsilon = 1/N$ to get the result for the \oedp.
\begin{theorem}\thmlabel{immediate-MG}
Given a stream $S$ of size $\Ns$ and parameters $\epsilon$ and $\phi$, where $1/\Ns \leq \epsilon < \phi< 1$ and 
$\phi > 1/M +\Omega(1/N)$,
the approximate \oedp can be solved 
at an  amortized   I/O complexity
$O\! \left( \left (\frac{1}{B} + \frac{M}{(\phi M - 1) \Ns} \right) 
 \log {\frac{1}{\epsilon M}} \right)$
per stream item.
\end{theorem}

\begin{proof}
  Correctness follows from the arguments above.  We need only analyze
  the I/O costs.  We analyze the I/O costs of the insertions and the
  queries separately.



  The amortized cost of performing insertions is
  $O(\frac{1}{B}\log\frac{1}{\epsilon M})$.

  

  To analyze the query costs, let $\epsilon_0 = 1/M$, i.e., the
  frequency-approximation error of the in-memory level of our data
  structure.
  
  Since we perform at most one query each time an item's count in
  $\calC_0$ goes from 0 to $(\phi - \epsilon_0)\Ns$, the total number
  of queries is at most $\Ns / ((\phi - \epsilon_0)\Ns)=1/(\phi -
  \epsilon_0) = M/(\phi M - 1)$.  Since each query costs $O(\log
  (1/\epsilon M))$ I/Os, the overall amortized I/O complexity of the
  queries is $O\! \left( \left (\frac{M}{(\phi M - 1) \Ns} \right)
  \log {\frac{1}{\epsilon M}} \right)$.
%
\end{proof}

\pparagraph{Exact reporting}  If no false positives
are allowed, we set $\epsilon = 1/\Ns$ in~\thmref{immediate-MG}.  
For error-free reporting, we must store all the items, which
increases the number of levels and thus the I/O cost.  In particular,
we have the following result on \oedp.

\begin{corollary}\corlabel{immediate-noerrors}
Given a stream $S$ of size $\Ns$ and $\phi > 1/M +\Omega(1/N)$  
%
the \oedp can be solved 
at amortized I/O complexity
  $O\! \left( \left (\frac{1}{B} + \frac{M}{(\phi M - 1) \Ns} \right)
    \log  {\frac{N}{M}} \right)$ per stream item.
\end{corollary}

\pparagraph{Summary} The external-memory MG algorithm supports
a throughput at least as fast as optimal write-optimized
dictionaries~\cite{BrodalFa03b,BenderFaJa15,BenderFaFi07,BrodalDeFi10,
  BuchsbaumGoVe00,BenderFa17}, while estimating the counts as well as
an enormous RAM.  It maintains count estimates at
different granularities across the levels. Not all estimates are actually needed for
each structure, but given a small number of levels, we can refine 
the count estimates by looking in only a few additional
locations.

The external-memory MG algorithm helps us solve the \oedp.  The smallest MG sketch (which fits in memory) is the most
important estimator here, because it serves to sparsify queries to the
rest of the structure. When such a query gets triggered, we need the
total counts from the remaining $\log \frac{N}{M}$ levels for the
(exact) online event-detection problem but only $\log
\frac{1}{\epsilon M}$ levels when approximate thresholds are
permitted.  In the next two sections, we exploit other advantages of
this cascading technique to support much lower $\phi$ without
sacrificing I/O efficiency.

%
%


 


\section{Event Detection With Time-Stretch}\seclabel{timestretch}
The external-memory Misra-Gries algorithm described in
\secref{immediate-MG} reports events immediately, albeit at a higher
amortized I/O cost for each stream item.  In this section, we show
that, by allowing a bounded delay in the reporting of events, we can
perform event detection asymptotically as cheaply as if we reported
all events only at the end of the stream.

\pparagraph{Time-stretch filter}
We design a new data structure to guarantee time-stretch called the
\defn{time-stretch filter}.  Recall that, in order to guarantee a
time-stretch of $\alpha$, we must report an item $x$ no later than
time $t_1+(1+\alpha)F_t$, where $t_1$ is the time of the
first occurrence of $x$, and $F_t$ is the flow time of $x$.

Similar to the external-memory MG structure, the time-stretch filter consists
of $L = \log_r({1}/{(\epsilon M)})$ levels
$\calC_0,\ldots,\calC_{L-1}$. The $i$th level has size $r^iM$.
Items are flushed from lower levels to higher levels.

Unlike the data structure in~\secref{immediate-MG} for the \oedp, all events are detected during 
the flush
operations.  Thus, we never need to perform point queries.  This
means that (1) we can use simple sorted arrays to represent each level
and, (2) we don't need to maintain the invariant that level 0 is a
Misra-Gries data structure on its own.



\pparagraph{Layout and flushing schedule} 
We split the table at each level $i$ into $q = (\alpha + 1)/\alpha$ equal-sized
\defn{bins} $b_1^i, \ldots, b_q^i$, 
each of size $\frac{\alpha}{\alpha + 1} (r^iM)$.
The capacity of a bin is defined by the sum of the counts of the items
in that bin, i.e., a bin at level $i$ can become full because it contains
$\frac{\alpha}{\alpha + 1} (r^iM)$ items, each with count 1, or 1 item
with count $\frac{\alpha}{\alpha + 1} (r^iM)$, or any other such combination.

We maintain a strict flushing schedule to obtain the time-stretch
guarantee.  The flushes are performed at the granularity of bins
(rather than entire levels).
Each stream item is inserted into $b_1^0$.
Whenever a bin $b_1^i$ becomes full (i.e., the sum of the counts of the
items in the bin is equal to its size), we shift all the bins on
level $i$ over by one (i.e., bin 1 becomes bin 2, bin 2 becomes bin 3,
etc), and we move all the items in $b_q^i$ into bin
$b_1^{i+1}$.
Since the bins in level $i+1$ are $r$ times larger than the bins in level $i$, bin $b_1^{i+1}$
becomes full after exactly $r$ flushes from  $b_q^i$.
When this happens, we perform a flush on level $i+1$ and so on.  Starting from the beginning, every
$r^{i-1}M$ elements from the stream
causes a flush that involves level $i$.

Finally, during a flush involving levels $0, \ldots, i$, where $i \leq L-1$, we scan these levels 
and for each item $k$ in the input levels, we
sum the counts of each instance of $k$. If the total count is greater than $(\phi - \epsilon)
\Ns$, and (we have not reported it before) then we report\footnote{For each reported item, we set a flag 
that indicates it has been reported, to avoid duplicate reporting of events.} $k$.


\pparagraph{Correctness} We first prove correctness of the time-stretch filter.

\begin{lemma}\lemlabel{ts-correct}
The time-stretch filter reports each $\phi$-event $s_t$ occurring
at time $t$ at or before $t + \alpha F_t$, where $F_t$
is the flow-time of $s_t$.
\end{lemma}

\begin{proof}
In the time-stretch filter, each item inserted at level $i$ waits in
$1/\alpha$ bins until it reaches the last bin,
that is, it waits at least $r^i/\alpha$ flushes (from main memory)
before it is moved down to level $i+1$. This ensures that
items that are placed on a deeper level have aged sufficiently
that we can afford to not see them again for a while.


Consider an item $s_t$  
with flow time $F_t = t - t_1$, where $t$ is a $\phi$-event and $t_1$ is the time step of the
first occurrence of $s_t$. 

Let $\ell \in \{0,1,\ldots, L\}$ be the largest level containing an instance of $s_t$ at time $t$, when
$s_t$ has its $\phi \Ns$th occurrence. The flushing schedule guarantees that
the item $s_t$ must have survived at least $r^{\ell-1}/\alpha$ flushes since it was
first inserted in the data structure. Thus, $r^{\ell-1} M/\alpha \leq F_t$.

Furthermore, level $\ell$ is involved in a flush again
after $t_\ell = r^{\ell-1}M \leq \alpha F_t$ time steps. 
At time $t_{\ell}$ during the flush all counts of the item will be consolidated
to a total count estimate of $\tilde c$. Note that $\ell \leq L$
and the count-estimate error of $s_t$ can be at most $\epsilon N_{t_{\ell}}$, where
$N_{t_\ell}$ is the number of the stream items seen up till $t_\ell$.
Thus, we have that $\phi \Ns \tilde c + \epsilon N_t
\leq \tilde c + \epsilon \Ns$. That is, 
$\tilde c \geq (\phi - \epsilon) \Ns$, which means that $s_t$ gets reported during the flush at time $t_{\ell}$,
which is at most $\alpha F_t$ time steps away from $t_1$. 
\end{proof}

\pparagraph{I/O complexity} 
Next, we analyze the I/O complexity of the time-stretch filter.
We treat each level of the filter as a sorted array.

\begin{theorem}\thmlabel{timestretch}
Given a stream $S$ of size $\Ns$ and parameters $\epsilon$ and $\phi$,
where
$1/N \leq \epsilon < \phi < 1$, the approximate \oedp 
can be solved with time-stretch $1+\alpha$ at an amortized I/O
complexity 
$O(\frac{\alpha+1}{\alpha}(\frac{1}{B} \log \frac{1}{\epsilon M}))$
per stream item.
\end{theorem}

\begin{proof}

	A flush from level $i$ to $i+1$ costs $O(\frac{r^{i+1}M}{B})$ I/Os,
and moves $\frac{\alpha}{\alpha+1} r^i M$ stream items down one level, so the amortized cost
  to move one stream item down one level is~$O(\frac{r^{i+1}M}{B} / \frac{\alpha}{\alpha+1} r^i M) = 
O(\frac{\alpha+1}{\alpha} \frac rB)$~I/Os.
  
  Each stream item can be moved down at most $L$ levels, thus the
  overall amortized I/O cost of an insert is $O(\frac{\alpha+1}{\alpha} \frac {rL}{B}) 
= O\left( \frac{\alpha+1}{\alpha} \frac rB  \log_r \frac{1}{\epsilon M} \right)$, which is minimized at $r = e$.
%
%
\end{proof}

\pparagraph{Exact reporting with time-stretch}
Similar to~\secref{immediate-MG}, if we do not want any false positives among the reported events, we set $\epsilon = 1/\Ns$.
The cost of error-free reporting is that we have to store all the items,
which increases the number of levels and thus the I/O cost.
In particular, we have the following result on \oedp.

\begin{corollary}\corlabel{timestretch-noerror}
Given $\alpha>0$ and a stream $S$ of size $\Ns$, the \oedp can be
  solved with time stretch $1+\alpha$ at an amortized cost of
  $O\left(\frac{\alpha+1}{\alpha}\frac{\log (N/M)}{B}\right)$ I/Os per stream item.
\end{corollary}

\pparagraph{Summary} By allowing a little delay, we can solve the
timely event-detection problem at the same asymptotic cost as simply
indexing our
data~\cite{BrodalFa03b,BenderFaJa15,BenderFaFi07,BrodalDeFi10,
  BuchsbaumGoVe00,BenderFa17}.  

Recall that in the online solution the increments and decrements 
of the MG algorithm determined the flushes from one level
to the other. In contrast, 
these flushing decisions in the time-stretch solution were based entirely on the age of the items. The MG style count estimates
came essentially for free from the size and cascading nature of the levels.
%
Thus, we get different
reporting guarantees depending on whether we flush based on age or count.  

Finally, our results on \oedp and \oedp with time stretch 
show that there is a spectrum between completely online and completely
offline, and it is tunable with little I/O cost.


    




\section{Power-Law Distributions}\seclabel{powerlaw}

In this section, we present a data structure that
solves the \oedp on streams where the count of items follow a power-law distribution.
There is no assumption on the order of arrivals, which can be adversarial. 
In contrast to worst-case count distributions, 
our data structure for power-law inputs can support smaller reporting thresholds and achieve 
better I/O performance.


We note that previous work has analyzed the performance of Misra-Gries
style algorithms on similar input distributions. 
In particular, Berinde et al.~\cite{BerindeIn10}
consider streams where the item counts follow a Zipfian distribution,
the assumptions of which are similar but distinct from power-law.

Next, we briefly review the distinction and relationship between Zipfian
and power-law distributions. This will allow us to compare Berinde et al.'s result to our work.
For detailed review of these distributions, see~\cite{newman2005power,
ClausetShNe09, BreslauCa99, adamic2008zipf}.

%
%

\pparagraph{Zipfian vs. power-law distributions}
Let $f_1, \ldots, f_{u}$ be the ranked-frequency vector, that is, $f_1 \geq f_2
\geq \ldots \geq f_{u}$ of $u$ distinct items in a
stream of size $N$, where $u=|\calU|$. 
The item counts in the stream
follow a
\defn{Zipfian distribution} with exponent $\alpha> 0$
if frequency $f_i =\mathcal{Z} \cdot i^{-\alpha}$, where
$\mathcal{Z}$ is the normalization constant. In contrast, the item counts in the stream
follow a power-law distribution with exponent $\theta >1$ 
if the probability that an item has count $c$ is equal to ${Z} \cdot c^{-\theta}$, 
where $Z$ is the normalization constant.

An stream follows a Zipfian distribution with exponent $\alpha$ if and only if
it follows a power-law distribution with exponent $\theta = 1+1/\alpha$; see~\cite{adamic2008zipf} for
details on this conversion.

Berinde et al.~\cite{BerindeIn10} show that if the item counts
in the stream follow a Zipfian distribution with $\alpha>1$, 
then the MG algorithm can solve the $\epsilon$-approximate heavy
hitter problem using only $\epsilon^{-1/\alpha}$ words. 
Alternatively, on such Zipfian distributions, the MG algorithm achieves an improved error bound $\epsilon^{\alpha}$ using $1/\epsilon$ words. Since all our algorithm so far use the MG algorithm as
a building block, we automatically achieve these improved bounds for Zipf exponents $\alpha >1$ (that is,
power-law exponents $\theta \leq 2$).


However, many common power-law distributions found in nature
have $2 \leq \theta \leq 3$~\cite{newman2005power}.

In this section, we design a new external-memory data structure for the \oedp
with improved guarantees when the power-law exponent $\theta \geq 2 + 1/(\log_2 N/M)$.
%


\pparagraph{Preliminaries}
We use the continuous power-law definition\cite{newman2005power}:
the count of an item with a {power-law distribution}
has a probability $p(x) \,dx$ of taking a value in the
interval from $x$ to $x+ dx$, where $p(x) = Z \cdot x^{-\theta}$,
where $\theta>1$ and $Z$ is the normalization constant. 

In general, the power-law distribution on $x$ may
hold above some minimum value $c_{\min}$ of $x$.  For simplicity, we let
$c_{\min}=1$. The normalization constant $Z$ is calculated as follows.

\begin{align*}
 1= \int_{1}^{\infty}p(x)\,dx = Z \int_{1}^{\infty} x^{-\theta} \, dx = \frac{Z}{\theta-1} \Big[\frac{-1}{x^{\theta-1}} \Big]_{1}^{\infty} = \frac{Z}{\theta-1}. 
\end{align*}
Thus, $Z = (\theta-1)$.\footnote{In principle, one could have power-law distributions with $\theta<1$, but these distributions cannot be normalized 
and are not common~\cite{newman2005power}.} 
%
%
We will use the cumulative distribution of a power law, that is,
\begin{align}
\mbox{Prob } (x>c) &=\int_{j=c}^{\infty} \mbox{Prob } (x=c) =  \int_{j=c}^{\infty} (\theta-1) x^{-\theta} dx \nonumber\\ \notag
&= \big[- x^{-\theta+1} \big]_{c}^{\infty} = \frac{1}{c^{\theta-1}}. \stepcounter{equation}\tag{\theequation}\label{eq:cumul}
\end{align}


\subsection{Power-law filter}
First, we present the layout of our data structure, the \defn{power-law filter} and then we 
present its main algorithm, the shuffle merge, and finally
we analyze its performance.
 
\pparagraph{Layout} 
The power-law filter consists of a cascade of Misra-Gries tables, where $M$ is
the size of the table in RAM and there are $L = \log_r (2/\epsilon M)$ levels on disk,
where the size of level $i$  is $2/(r^{L-i}\epsilon)$.

Each level on disk has an \emph{explicit upper
bound} on the number of instances of an item that can be stored on that level.
This is different from the MG algorithm, where this upper bound is implicit: based on the level's size. In particular,
each level $i$ in the power-law filter has a \defn{level threshold} $\tau_i$ for $1 \leq i \leq L$,
($\tau_1 \ge  \tau_2  \ge  \ldots \ge \tau_L$), 
indicating that the maximum count on level $i$ can be $\tau_{i}$.

\pparagraph{Threshold invariant} We maintain the
invariant that at most $\tau_i$ instances of an item can be stored on level $i$. 
Later, we show how to set $\tau_i$'s based on the item-count distribution.

\pparagraph{Shuffle merge} The external-memory MG data structure and time-stretch filter use two different flushing strategies, and here we 
present a third for the power-law filter.


The level in RAM receives inputs from the stream one at a time. When attempting to insert to a level $i$ that is at capacity, instead
of flushing items to the next level, we find the smallest level $j > i$, which has enough empty space to hold all items from levels $0, 1, \ldots, i$. 
We aggregate the count of each item $k$ on levels $0, \ldots, j$, resulting
in a consolidated count $c_k^j$. 
If $c_k^j\geq (\phi - \epsilon) \Ns$, we report $k$. Otherwise, we pack instances
of $k$ in a bottom-up fashion on levels $j, \ldots, 0$%
, while maintaining
the threshold invariants. In particular,
we place $\min \{c_k^j, \tau_j \}$ instances of $k$ on level $j$,
and 
$\min \{c_k^j - (\sum_{\ell = y+1}^{j} \tau_y), \tau_y\}$ instances of $k$ on level $y$ for
$0 \leq y \leq j-1$.

Thus, the threshold invariant prevents us from flushing too many counts of an item
downstream. As a result, items get \defn{pinned}, that is, they cannot be flushed
out of a level. Specifically, we say an item is \defn{pinned at level $\ell$} if its
count exceeds $\sum_{i = L}^{\ell+1} \tau_i$. 

Too many pinned items at a level can clog the data structure. 
In~\lemref{powerlaw-noclog}, we show that if the item counts in the stream follow a
power-law distribution with exponent $\theta$, we can set the thresholds
based on $\theta$ in a way that no level has
too  many pinned items. 

\pparagraph{Online event detection}
As soon as the count of an item $k$ in RAM (level $0$) reaches a threshold of
$\phi \Ns - 2\tau_1$, the data structure triggers a sweep of all the $L$ levels, consolidating the count estimates
of $k$ at all levels. If the consolidated count reaches $(\phi  - \epsilon)\Ns$, we report $k$; otherwise we update the $k$'s consolidated count in RAM and ``pin'' $k$ in RAM,
that is, mark a bit to ensure $k$ does not participate in future shuffle merges. 
Reported items are remembered, so that each event gets reported exactly once.

\pparagraph{Setting thresholds} We now show how to set the level thresholds based
on the power-law exponent so that the data structure does not get ``clogged'' even though the high-frequency items 
are being sent to higher levels of the data structure.

\begin{lemma}\lemlabel{powerlaw-noclog} Let the item counts in an stream $S$
of size $\Ns$ be drawn from a power-law distribution
with exponent $\theta > 1$.  Let $\tau_{i} = r^{\frac{1}{\theta-1}} \tau_{i+1}$ for $1 \leq i \leq L-1$
and $\tau_L = (r\epsilon \Ns)^{\frac {1}{\theta-1}}$.
Then the number of keys pinned at any level $i$ is at most half its size, i.e., $1/(r^{L-i}\epsilon)$. 
\end{lemma}
\begin{proof}
We prove by induction on the
number of levels.  We start at level $L-1$.  An item is placed at level $L-1$ if its count is greater than
$\tau_L = (r\epsilon \Ns)^{\frac {1}{\theta-1}}$.  
By Equation~(\ref{eq:cumul}), there can be at most $\Ns/ \tau_L^{\theta-1}= 
\Ns/ (r \epsilon \Ns) = 1/r \epsilon$
such items which proves the base case.

Now suppose the lemma holds for level $i+1$. We show that it holds for
level $i$.  An item gets pinned at level ${i+1}$ if its count is
greater than $\sum_{\ell=L}^{i+2} \tau_{\ell}$. 

Using Equation~(\ref{eq:cumul}) again, the expected number of such items 
is  
\[ \leq \frac{\Ns}{{(\sum_{\ell=L}^{i+2} \tau_{\ell})}^{\theta-1}}  < \frac{{\Ns}}{{\tau_{i+2}}^{\theta-1}}.\]
By the induction hypothesis, this is at most half the size of
level $i+1$, that is,  
\[ \frac{{\Ns}}{{\tau_{i+2}}^{\theta-1}} \leq  \frac{1}{\epsilon {r^{L-i-1}}}.\]

Using this, we prove that the expected number of items pinned at level $i$ is at most $1/(r^{L-i}\epsilon$.

The expected number of pinned items at level $i$ is
\begin{align*}
\frac{\Ns}{{(\sum_{\ell=L}^{i+1}\tau_{\ell})}^{\theta-1}} &< \frac{\Ns}{({{\tau_{i+1}}^{\theta-1}})}= \frac{\Ns}{{{(r^{1/\theta-1} \cdot \tau_{i+1})}^{\theta-1}}} \\
&=  \frac 1r \cdot \frac{\Ns}{({{
\tau_{i+2}}^{\theta-1}})} \leq \frac{1}{ r \epsilon {r^{L-i-1}}} = \frac{1}{r^{L-i}\epsilon}.
\end{align*} 
\end{proof}

\subsection{Analysis}
Next, we prove correctness of the power-law filter
and analyze its I/O complexity.

We first establish notation.
Let $S$ be the stream of size $N$ where
the count of items follow a power-law
distribution with exponent $\theta >1$.
For simplification we use
 $\gamma = 2 ({N}/{M})^{\frac{1}{\theta-1}}$
in the analysis. 

\pparagraph{Correctness}
Next, we prove that the power-law filter reports all $\phi$-events
as soon as they occur. In the approximate \oedp, it may report
false positives, that is, items with frequency between
$(\phi - \epsilon)N$ and $\phi N$. As before, for error-free
reporting we set $\epsilon = 1/N$.

\begin{lemma}
The power-law filter solves the approximate \oedp on $S$.
\end{lemma}

\begin{proof}
Let $\tilde{c}_i$ denote the count estimate of an item $i$ in RAM in the power-law
filter. Let $f_i$ be the frequency of $i$ in the stream. 
Since at most $\sum_{\ell = L}^1 \tau_\ell < 2\tau_1$
instances of a key can be stored on disk, we have that: $\tilde{c}_i \leq f_i \leq \tilde{c}_i + 2 \tau_1$.

Suppose item $s_t$ reaches the threshold $\phi \Ns$ at time $t$, then its
count estimate $s_t$ in RAM must be at least $\tilde{c}_i \geq \phi \Ns - 2\tau_1 = \phi N- 2r^{L/\theta-1} (\epsilon N)^{1/\theta-1}
= \phi N- 2(N/M)^{\frac{1}{\theta-1}} = \phi N  - \gamma$.
This is exactly when we trigger a sweep of the 
data structure consolidating the count of $s_t$ across all $L$ levels;
if the consolidated count reaches $(\phi N - \epsilon) \Ns$, we report it. This proves
correctness as the consolidated count can 
have an error of at most $\epsilon \Ns$.
\end{proof}

\pparagraph{I/O complexity}
We now analyze the I/O complexity of the power-law filter.
Similar to~\secref{immediate-MG}, we assume each level is implemented as a
B-tree, although the same basic algorithm works with sorted arrays
(included with fractional cascading from one level to the next, similar to
cache-oblivious lookahead arrays~\cite{BenderFaFi07}).

\begin{theorem}\thmlabel{power-law}
  Let $S$ be a stream of size $\Ns$ where the count of items follow  a power-law distribution with exponent $\theta >1$.  Let
  $\gamma = 2 \bigl(\frac{N}{M}\bigr)^{\frac{1}{\theta-1}}$.  
Given $S$, $\epsilon$ and $\phi$,
such that $1/\Ns \leq \epsilon <\phi$ and $\phi = \Omega(\gamma/N)$, the
  approximate \oedp can be solved at an amortized I/O
  complexity
  $O\! \left( \left(\frac{1}{B} + \frac{1}{{(\phi N - \gamma)}^{\theta-1}} \right)  \log \frac {1}{\epsilon M} \right)$ per stream item.
\end{theorem}

\begin{proof}
The insertions cost $O(rL/B)  = O((1/r) \log_r (1/\epsilon M))$
as we are always able to flush out a constant fraction of a level during
a shuffle merge using~\lemref{powerlaw-noclog}. This cost is minimized at $r=e$.

Since we perform at most one query each time an item's count in
  RAM reaches $(\phi N - \gamma)$. The total number
  of items in the stream with count at least $(\phi N - \gamma)$ is at most $\Ns / {(\phi N - \gamma)}^{\theta-1}$.  
Since each query costs $O(\log
  (1/\epsilon M))$ I/Os, the overall amortized I/O complexity of the
  queries is $O\! \left( \frac{1}{{(\phi N - \gamma)}^{\theta-1}} 
  \log {\frac{1}{\epsilon M}} \right)$.
%
\end{proof}

\pparagraph{Exact reporting}
To forbid false positives, we set $\epsilon = 1/\Ns$ and get the following
corollary.

\begin{corollary}\corlabel{oedp-power-law}
  Let $S$ be a stream of size $\Ns$ where the count of items follow a power-law distribution with exponent $\theta > 1$.  Let
  $\gamma = 2 \bigl(\frac{N}{M}\bigr)^{\frac{1}{\theta-1}}$.  Given 
$\phi = \Omega(\gamma/N)$, the \oedp can be solved at an amortized I/O
  complexity
  $O\! \left( \left(\frac{1}{B} + \frac{1}{{(\phi N - \gamma)}^{\theta-1}} \right) \log \frac {N}{ M} \right)$ per stream item.
\end{corollary}

\pparagraph{Remark on scalability}
Notice that the power-filter on an stream with a power-law distribution allows for strictly smaller thresholds $\phi$ compared to~\thmref{immediate-MG}
and~\corref{immediate-noerrors}
on worst-case-distributions, when $\theta > 2 + 1/(\log_2 (N/M))$.
Recall that we need $\phi \geq \Omega(1/M)$ for solving \oedp on worst-case streams.
In contrast, in~\thmref{power-law} and~\corref{oedp-power-law}, we need $\phi \geq \Omega(\gamma/N)$.
When we have a power-law distribution with $\theta \geq 2 + 1/(\log_2 N)$,
we have $\frac \gamma N =  \frac{2}{M^{1/(\theta-1)}N^{\theta-2}} < \frac 1M$ for $\theta \geq 2 + 1/(\log_2 (N/M))$.
%
%



\pparagraph{Remark on dynamic thresholds}\label{dynamic_thresholds} 
Finally, we
argue that 
level thresholds of the power-law filter can be set dynamically
when the power-law exponent $\theta$ is not known ahead of time. 

Initially, each level on disk has a threshold $0$
(i.e., $\forall i\in{1, \ldots, L}\ \tau_i = 0$). During the first shuffle-merge
involving RAM and the first level on disk, we determine the minimum threshold
for level $1$ ($\tau_1$) required in-order to move at least half of the items
from RAM to the first level on disk.
%
%
When multiple levels, $0, 1, \ldots, i$, are involved in a shuffle-merge,  we
use a bottom-up strategy to assign thresholds. We determine the minimum
threshold required for the bottom most level involved in the shuffle-merge
($\tau_i$) to flush at least half the items from the level just above it
($\tau_{i-1}$). We then apply the same strategy to increment thresholds for
levels $i-1,\ldots, 1$.

This means that the $\tau_i$s for levels $1,\ldots, L$ increase
monotonically. Moreover, during shuffle-merges, we increase thresholds
of levels involved in the shuffle-merge from bottom-up and to the
minimum value so as to not clog the data structure, which means that
the $\tau_i$s take their minimum possible values. Thus, if the $\tau_i$ have
a feasible setting, then this adaptive strategy will find it.

\pparagraph{Summary}  With a power law distribution, we can support a
much lower threshold $\phi$ for the online event-detection
problem.
In the external-member MG sketch from \secref{external-MG-subsec}, the
upper bounds on the counts at each level are implicit. In this
section, we can get better estimates by making these bounds
explicit. Moreover, the data structure can learn these bounds
adaptively.  Thus, the data structure can automatically tailor itself
to the power law exponent without needing to be told the exponent explicitly.

    







\section{Motivating National Security Application}
\seclabel{application}

In this section, we describe the more complex national-security
setting that motivates our constraints.  We describe Firehose~\cite{FirehoseSite,AndersonPl15}, a clean
benchmark that captures the fundamental elements of this setting.  The
\oedp in this paper in turn distills the most difficult part of the
Firehose benchmark.  Therefore our solutions have direct line of sight
to important national-security applications.

\begin{figure}
\centering
\includegraphics[width=5in]{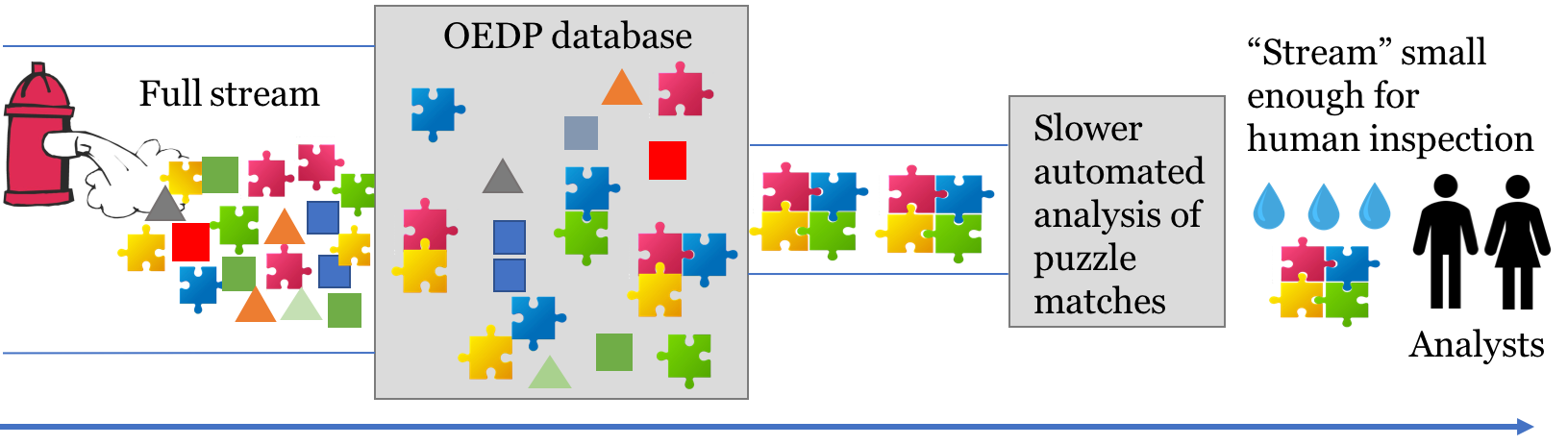}
\caption{The analysis pipeline that motivates our OEDP solution. Analysts associate a multi-piece pattern, represented by the 4-piece
puzzle, to a high-consequence event. The pieces arrive slowly over time, mixed with innocent traffic
in a high-throughput ``firehose'' stream. Our database stores many partial matches to the pattern reporting all instances of the pattern.  
There still may be fair number of matches, which are pared down by an automated system to a small number (essentially droplets
compared to the original stream) of matches worthy of human inspection.
}
\label{fig:pipeline}
\end{figure}

We are motivated by monitoring systems for national security~\cite{FirehoseSite,AndersonPl15},
where experts associate special patterns in a cyberstream to
rare, high-consequence real-life events. These patterns are formed by a small number of ``puzzle pieces,'' as shown in~\figref{pipeline}. Each
piece is associated with a key such as an IP address or a hostname.  The pieces arrive over time. When an entire puzzle associated
with a particular key is complete, this is an event, which should be reported as soon as the final puzzle piece falls into place.
In~\figref{pipeline}, the first stage is like our \oedp algorithm, except that it must store puzzle pieces with each key rather than a count and
the reporting trigger is a complete puzzle, not a count threshold.

There can still be a fair number of matches to this special pattern, most of which are still not the critically bad event.  This might overwhelm a human analyst, who would then not use the system.
However, automated tools, shown in the second stage of~\figref{pipeline}, can pare these down to the few
events worthy of analyst attention.

The first stage filter, like our \oedp solution, must struggle to
handle a massively large, fast stream.  It is reasonable to allow a
few false positives in the first stage to improve its speed. The
second stage can screen out almost all of these false positives as
long as the stream is significantly reduced.  The second stage is a
slower, more careful tool which cannot keep up with the initial
stream.  This second tool cannot, however, repair false negatives
since anything the first filter misses is gone forever. So the first
tool cannot drop any matches to the pattern. Experts have gone to
great effort to find a pattern that is a good filter for the
high-consequence events.  We do not allow false negatives because the
high-consequence events that match this carefully crafted pattern can
and must be detected.

Each of these patterns are small with respect to the stream size, so the detection algorithm must be scalable, that is,
must be able to support a
small $\phi$.  The consequences of missing an event (false negative) are so severe that it is not reasonable to
risk facing those consequences just to save a little space.
Thus we must save all partial patterns, motivating our use of external memory.

The DoD Firehose benchmark captures the essence of this setting~\cite{FirehoseSite}.  In Firehose, the input stream has (key,value) pairs.  When a key is seen for the 24th time, the system must return
a function of the associated 24 values.  The most difficult part of this is determining when the 24th instance of a key arrives.  Thus like Firehose, the \oedp captures the essence of the
motivating application. 

\section{Conclusion}
Our results show that, by enlisting
the power of external memory, we can solve online event detection problems at
a level of precision that is not possible in the streaming model, and
with little or no sacrifice in terms of the timeliness of reports.

Even
though streaming algorithms, such as Misra-Gries, were 
developed for a space-constrained setting, they are nonetheless useful
in external memory, where storage is plentiful but I/Os are expensive.  
%
Furthermore, using external memory for
problems that have traditionally been 
analyzed in the streaming setting enables solutions that can scale
beyond the provable limits of fast RAM

\subsection*{Acknowledgments}

We would like to thank Tyler Mayer for many helpful discussions in earlier
stages of this project.  In~\figref{pipeline}, the full-puzzle icon is from \url{theme4press.com}, the fire-hydrant
icon is from \url{https://hanslodge.com} and the water-drop icon is from \url{stockio.com}.


\bibliographystyle{abbrv}
\raggedright\newcommand{\noopsort}[1]{} \newcommand{\singleletter}[1]{#1}
  \punt{ Uncomment the following lines for short conference/journal names
  @String{SODA = {SODA}} @String{JACM = {Journal of the ACM}} @String{SPAA =
  {SPAA}} @String{PPoPP = {PPoPP}} @String{PLDI = {PLDI}} @String{STOC =
  {STOC}} @String{FOCS = {FOCS}} @String{ESA = {ESA}} @String{ALP = {Colloquium
  on Automata, Languages, and Programming}} @String{SWAT = {SWAT}}
  @String{JALGO = {Journal of Algorithms}} @String{PODC = {PODC}} @String{LNCS
  = {LNCS}} @String{SUPERCOMP = {Supercomputing}} @String{ICCSE = {Israeli
  Conference on Computer Systems Engineering}} @String{CMD = {Conference on
  Management of Data}} }

\end{document}